\documentclass[11pt]{rspublic}%
\usepackage{amsfonts}
\usepackage{amsmath}
\usepackage{fullpage}
\usepackage{amssymb}
\usepackage{graphicx}
\usepackage{natbib}
\providecommand{\U}[1]{\protect\rule{.1in}{.1in}}
\newtheorem{example}[theorem]{Example}
\newtheorem{remark}[theorem]{Remark}
\begin{document}

\let\originalleft\left
\let\originalright\right
\renewcommand{\left}{\mathopen{}\mathclose\bgroup\originalleft}
\renewcommand{\right}{\aftergroup\egroup\originalright}

\title[Sequential decoding of a general classical-quantum channel]{Sequential decoding of a general\\classical-quantum channel}
\author{Mark M. Wilde}
\affiliation{School of Computer Science, McGill University, Montreal, Quebec H3A 2A7, Canada}
\label{firstpage}
\maketitle

\begin{abstract}
{sequential decoding, hypothesis testing relative entropy, non-commutative union bound, Naimark extension theorem, quantum polar codes}
Since a quantum measurement generally disturbs the state of a
quantum system, one might think that it should not be possible for a sender and
receiver to communicate reliably when the receiver performs a
large number of sequential measurements to determine the message of the
sender. We show here that this intuition is not true, by demonstrating that
a sequential decoding strategy works well even in the most general \textquotedblleft
one-shot\textquotedblright\ regime, where we are given a single instance of a
channel and wish to determine the maximal number of bits that can be
communicated up to a small failure probability. This result follows
by generalizing a
non-commutative union bound to apply for a sequence
of general measurements.
We also demonstrate two
ways in which a receiver can recover a state close to the original state after
it has been decoded by a sequence of measurements that each succeed with high
probability. The second of these methods will be useful in realizing an
efficient decoder for fully quantum polar codes, should a method ever be found
to realize an efficient decoder for classical-quantum polar codes.

\end{abstract}

\section{Introduction}

\label{sec:intro}The reliable communication of classical data over quantum
channels is one of the earliest problems to be considered in quantum
information theory. Some of the most important contributions to this problem
(to name just a few) are the Holevo upper bound on the accessible information
\citep{Holevo73}, the coding theorem due to \citep{Hol98,PhysRevA.56.131} (HSW), and the fact that entangled
signaling states can enhance communication rates for certain quantum channels
\citep{H09}.

The main difference between the proofs of the HSW\ theorem and Shannon's
classical channel capacity theorem \citep{bell1948shannon}\ is that, in the
former case, one has to specify a quantum measurement that recovers the
classical data being transmitted (a quantum decoder), as opposed to a
classical algorithm that does so. Indeed, in their respective proofs,
HSW\ demonstrated that a quantum measurement known as the \textquotedblleft
pretty-good\textquotedblright\ or \textquotedblleft
square-root\textquotedblright\ measurement \citep{B75,B75a,PhysRevA.54.1869}\ allows a
receiver to decode classical information reliably at a rate equal to the
Holevo rate. For a pure-loss bosonic channel modeling free-space
communication, for example, this Holevo rate can be significantly higher than
data rates that are achievable with more traditional measurement strategies
such as homodyne or heterodyne detection~\citep{GGLMSY04}.

As in the HSW\ decoding measurement, we typically perform measurements on
quantum systems in order to gain information about them, and one well-known
feature of quantum mechanics is that a measurement can disturb the state of
the system that we are measuring. Thus, it came as a surprise when Lloyd,
Giovannetti, and Maccone (LGM)
\citep{PhysRevLett.106.250501,PhysRevA.85.012302}\ showed that it is possible
to achieve the Holevo rate by performing independent sequential measurements,
in analogy with classical sequential decoding strategies \citep{book1991cover}.
A sequential decoding scheme proceeds according to the following simple algorithm:

\begin{enumerate}
\item Let $M$ be the total number of codewords. Initialize a counter $i=1$.

\item Perform a quantum measurement to determine if the transmitted codeword
is the $i^{\text{th}}$ codeword.

\item If the measurement result is \textquotedblleft yes,\textquotedblright%
\ decode as codeword $i$ and conclude. If the measurement result is
\textquotedblleft no,\textquotedblright\ increment $i.$

\item If $i\leq M$, go to step~2. Otherwise, declare failure.
\end{enumerate}

After the work of LGM, Sen presented a remarkable simplification of their
error analysis \citep{S11}, by establishing a non-commutative union bound that
holds for a set of projective measurements applied sequentially to a quantum
state (where the projective measurements do not necessarily commute). This
non-commutative union bound is an extension of the familiar union bound from
probability theory, and as such, it should find wide application in settings
beyond those considered in quantum communication theory. Sen applied his
non-commutative union bound to a variety of problems in \citep{S11}, including
the problem of classical communication over quantum channels, and it has since
been applied in designing Holevo-rate-achieving polar codes for
classical-quantum channels \citep{WG11}\ and in demonstrating how to decode the
pure-loss bosonic channel at the Holevo rate~\citep{WGTS11}.

All of the above results apply to a setting in which the channel is memoryless
and identically distributed, so that one use of it does not depend on the
others and so that each use leads to the same noise at the output as the other
uses, respectively. Given that this \textquotedblleft IID\textquotedblright%
\ setting is really just an idealization, there has been a strong effort to
develop a theory of quantum information that goes beyond the IID\ setting and
applies to channels with no structure whatsoever \citep{Renner2005,DH11,T12}.
This regime beyond the IID\ setting is known as the \textquotedblleft
one-shot\textquotedblright\ regime, where we are concerned with a single
instance of a resource and desire to make the best use of it up to some
controllable failure probability. In this vein, there have been several
contributions characterizing the reliable communication of classical data over
quantum channels \citep{HN03,mosonyi:072104,WR12}, and all of these employed
the \textquotedblleft pretty-good\textquotedblright\ measurement as the decoder.

Many of the developments listed above have improved our understanding of
classical communication over quantum channels, but there are some important
considerations left unanswered:

\begin{enumerate}
\item We know very well that the most general kind of measurement allowed in
quantum mechanics is a positive operator-valued measure (POVM). Does Sen's
bound generalize so that it applies for a sequence of general measurements?

\item Does sequential decoding work well in the one-shot regime?

\item When can one conclude that the state resulting from a sequence of
general measurements is close to the state before this sequence of
measurements occurs?
\end{enumerate}

This paper resolves the above problems, by showing that

\begin{enumerate}
\item Sen's non-commutative union bound applies not just for a sequence of
projections, but for the more general case of a sequence of positive operators
each with spectrum less than one. This result follows simply by applying the
well-known Naimark extension theorem.\footnote{This observation is due to
Andreas Winter and Aram Harrow from a discussion in December 2011 at
QIP~2012.} Thus, the non-commutative union bound now applies for a sequence of
general measurements (POVMs) and, as such, it should find wide application
in other areas of quantum information science.

\item Indeed, sequential decoding works well even in the one-shot regime. That
is, one can give a meaningful bound on the amount of information that can be
transmitted up to a failure probability no larger than $\varepsilon$ for some
$\varepsilon>0$ when using a sequential decoding strategy. The information
bound we present is very similar to the bound of \citep{WR12}.

\item A sequence of measurements followed by the reverse sequence of these
measurements causes only a negligible disturbance to a state if the original
sequence of measurements has a high probability of success. This last result
generalizes Winter's gentle measurement lemma \citep{itit1999winter} to the
more general setting of a sequence of measurements. One application of this
last result is in decoding fully quantum polar codes for arbitrary quantum
channels \citep{WR12a}.\footnote{Quantum polar codes are the only known
near-explicit quantum error-correcting codes that achieve the coherent
information rate of an arbitrary quantum channel.}
\end{enumerate}

We structure this paper as follows. The next section reviews some background
material, including the definition of the hypothesis testing relative entropy
\citep{BD10,WR12}, the Naimark extension theorem, and Sen's non-commutative union
bound \citep{S11}. We then proceed in the order given above.

\section{Review}

\subsection{Hypothesis testing relative entropy}

The hypothesis testing relative entropy, denoted as $D^{\varepsilon}_H\left(
\rho||\sigma\right)  $, is an entropy measure derived from the error
probabilities arising from a quantum measurement that attempts to distinguish
between the states $\rho$ and $\sigma$ (a quantum hypothesis test). The most
general measurement that one could use in such a test is a two-outcome
POVM\ $\left\{  Q,I-Q\right\}  $ where $0\leq Q\leq I$. The outcome $Q$
corresponds to deciding that the state is $\rho$ and the outcome $I-Q$
corresponds to deciding that the state is $\sigma$. Thus, the probability of
guessing correctly when the state is $\rho$ is equal to Tr$\left\{
Q\rho\right\}  $, and the probability of guessing incorrectly when the state
is $\sigma$ is equal to Tr$\left\{  Q\sigma\right\}  $. In an asymmetric
quantum hypothesis test, we try to find a POVM\ that guesses $\rho$ correctly
with high probability, so that%
\begin{equation}
\text{Tr}\left\{  Q\rho\right\}  \geq1-\varepsilon,
\end{equation}
for some small, fixed $\varepsilon\geq0$, while minimizing the probability
that we guess $\sigma$ incorrectly. This naturally leads to a semidefinite
optimization program, specified by the following quantity:%
\begin{equation}
\beta_{\varepsilon}\left(  \rho,\sigma\right)  \equiv\min_{Q}\left\{
\text{Tr}\left\{  Q\sigma\right\}  :0\leq Q\leq I,\ \text{Tr}\left\{
Q\rho\right\}  \geq1-\varepsilon\right\}  . \label{eq:beta-error}%
\end{equation}
By taking the negative logarithm of $\beta_{\varepsilon}\left(  \rho
,\sigma\right)  $, we arrive at the hypothesis testing relative entropy
defined in \citep{BD10,WR12}:%
\begin{equation}
D_{H}^{\varepsilon}\left(  \rho||\sigma\right)  \equiv-\log\beta_{\varepsilon
}\left(  \rho,\sigma\right)  . \label{eq:hypo-entropy}%
\end{equation}
One can derive other entropic measures based on the hypothesis testing
relative entropy that have various natural properties \citep{DKFRR12}.

\subsection{Naimark extension theorem}

We briefly review the Naimark extension theorem\ and a straightforward proof
of it. The importance of this theorem is that it demonstrates how one can
implement a general quantum measurement simply by performing a unitary on the
system of interest and a probe system, followed by a von Neumann measurement
of the probe.

\begin{theorem}
[Naimark]\label{thm:naimark}For any POVM $\left\{  \Gamma_{x}\right\}
_{x\in\mathcal{X}}$ acting on a system $S$, there exists a unitary $U_{SP}%
$\ (acting on the system $S$ and a probe system $P$) and an orthonormal basis
$\left\{  \left\vert x\right\rangle _{P}\right\}  _{x\in\mathcal{X}}$ such
that%
\begin{equation}
\text{\emph{Tr}}\left\{  U_{SP}^{\dag}\left(  I_{S}\otimes\left\vert
x\right\rangle \left\langle x\right\vert _{P}\right)  U_{SP}\left(  \rho
_{S}\otimes\left\vert 0\right\rangle \left\langle 0\right\vert _{P}\right)
\right\}  =\text{\emph{Tr}}\left\{  \Gamma_{x}\rho\right\}  .
\end{equation}

\end{theorem}

\begin{proof}
For every POVM $\left\{  \Gamma_{x}\right\}  $, we can form the following
isometry:%
\begin{equation}
V_{SP}\equiv\sum_{x}\left(  \sqrt{\Gamma_{x}}\right)  _{S}\otimes\left\vert
x\right\rangle \left\langle 0\right\vert _{P},
\end{equation}
which can be extended to a unitary operator $U_{SP}$ by appropriately filling
out the other $\left\vert \mathcal{X}\right\vert -1$ entries of the form:%
\begin{equation}
\sum_{x}\left(  A_{x,x^{\prime}}\right)  _{S}\otimes\left\vert x\right\rangle
\langle x^{\prime}|,
\end{equation}
for some operators $A_{x,x^{\prime}}$ and where $x^{\prime}\in\left\{
1,\ldots,\left\vert \mathcal{X}\right\vert -1\right\}  $. The statement of the
theorem then follows easily from this choice of unitary.
\end{proof}

\begin{example}
\label{ex:binary-POVM}Let $\left\{  \Gamma,I-\Gamma\right\}  $ be a binary
POVM acting on the system$~S$. Consider the following unitary operator
$U_{SP}$ acting on the system$~S$\ and a qubit probe system~$P$:%
\begin{equation}
U_{SP}\equiv\left(  \sqrt{\Gamma}\right)  _{S}\otimes\left\vert 0\right\rangle
\left\langle 0\right\vert _{P}+\left(  \sqrt{I-\Gamma}\right)  _{S}%
\otimes\left\vert 1\right\rangle \left\langle 0\right\vert _{P}-\left(
\sqrt{I-\Gamma}\right)  _{S}\otimes\left\vert 0\right\rangle \left\langle
1\right\vert _{P}+\left(  \sqrt{\Gamma}\right)  _{S}\otimes\left\vert
1\right\rangle \left\langle 1\right\vert _{P}.
\end{equation}
The above unitary corresponds to a Naimark extension of the POVM $\left\{
\Gamma,I-\Gamma\right\}  $.
\end{example}

\subsection{Non-commutative union bound}

This section recalls Sen's non-commutative union bound \citep{S11}. As we
mentioned in Section~\ref{sec:intro}, this bound should find wide application
in settings beyond those considered for communication, since it generalizes
the union bound from probability theory.

\begin{theorem}
[Sen]\label{thm:sen}For a subnormalized state $\sigma$ such that $\sigma\geq0$
and \emph{Tr}$\left\{  \sigma\right\}  \leq1$, and a sequence of Hermitian
projectors $\Pi_{1}$, \ldots, $\Pi_{M}$, the following non-commutative union
bound holds:%
\begin{equation}
\text{\emph{Tr}}\left\{  \sigma\right\}  -\text{\emph{Tr}}\left\{  \Pi
_{M}\cdots\Pi_{1}\sigma\Pi_{1}\cdots\Pi_{M}\right\}  \leq2\sqrt{\sum_{m=1}%
^{M}\text{\emph{Tr}}\left\{  \left(  I-\Pi_{m}\right)  \sigma\right\}  }.
\end{equation}

\end{theorem}

\section{Non-commutative union bound for POVMs}

We now give an extension of Sen's non-commutative union bound that applies for
general measurements.

\begin{lemma}
\label{lem:non-comm-bound}Let $\sigma$ be a subnormalized state such that
$\sigma\geq0$ and \emph{Tr}$\left\{  \sigma\right\}  \leq1$, and let
$\Lambda_{1}$, \ldots, $\Lambda_{M}$ denote a set of positive operators such
that $0\leq\Lambda_{m}\leq I$ for all $m\in\left\{  1,\ldots,M\right\}  $.
Then the following non-commutative union bound holds:%
\begin{equation}
\text{\emph{Tr}}\left\{  \sigma\right\}  -\text{\emph{Tr}}\left\{
\Pi_{\Lambda_{M}}\cdots\Pi_{\Lambda_{1}}\left(  \sigma\otimes\left\vert
\overline{0}\right\rangle \left\langle \overline{0}\right\vert _{P^{M}%
}\right)  \Pi_{\Lambda_{1}}\cdots\Pi_{\Lambda_{M}}\right\}  \leq2\sqrt
{\sum_{m=1}^{M}\text{\emph{Tr}}\left\{  \left(  I-\Lambda_{m}\right)
\sigma\right\}  }, \label{eq:ext-sens-bound}%
\end{equation}
where $\left\vert \overline{0}\right\rangle _{P^{M}}\equiv\left\vert
0\right\rangle _{P_{1}}\otimes\cdots\otimes\left\vert 0\right\rangle _{P_{M}}$
is an ancillary state of $M$ probe systems and $\Pi_{\Lambda_{i}}$ is a
projector defined as $\Pi_{\Lambda_{i}}\equiv U_{i}^{\dag}P_{i}U_{i}$, for
some unitary $U_{i}$ and projector $P_{i}$ such that%
\begin{equation}
\text{\emph{Tr}}\left\{  \Pi_{\Lambda_{i}}\left(  \sigma\otimes\left\vert
\overline{0}\right\rangle \left\langle \overline{0}\right\vert _{P^{M}%
}\right)  \right\}  =\text{\emph{Tr}}\left\{  \Lambda_{m}\sigma\right\}  .
\end{equation}

\end{lemma}

\begin{proof}
This extension of Sen's bound follows easily by employing the Naimark
extension theorem and Sen's non-commutative union bound.

To each POVM\ element $\Lambda_{i}$ (as in the statement of the theorem),
there exists a unitary $U_{SP_{i}}^{\left(  i\right)  }$ (acting on the system
$S$ and the $i^{\text{th}}$ probe system) and a projector $I_{S}%
\otimes\left\vert i\right\rangle \left\langle i\right\vert _{P_{i}}$ such that
the following relation holds%
\begin{equation}
\text{Tr}\left\{  \Pi_{\Lambda_{i}}\left(  \rho_{S}\otimes\left\vert
\overline{0}\right\rangle \left\langle \overline{0}\right\vert _{P^{M}%
}\right)  \right\}  =\text{Tr}\left\{  \Lambda_{i}\rho_{S}\right\}  ,
\end{equation}
where%
\begin{equation}
\Pi_{\Lambda_{i}}\equiv\left(  U_{SP_{i}}^{\left(  i\right)  }\right)  ^{\dag
}\left(  I_{S}\otimes\left\vert i\right\rangle \left\langle i\right\vert
_{P_{i}}\right)  U_{SP_{i}}^{\left(  i\right)  }.
\end{equation}
Observe that the operator $\Pi_{\Lambda_{i}}$ is a Hermitian projector,
so that Sen's bound applies to each
of these operators. Then (\ref{eq:ext-sens-bound}) follows from
Theorems~\ref{thm:sen}\ and \ref{thm:naimark}.
\end{proof}

\begin{remark}
Since Lemma~\ref{lem:non-comm-bound}\ applies for general measurements, it can
be used in the context of Sections~3 and 4 of \citep{S11}\ without the
need for constructing a particular kind of \textquotedblleft intersection
projector\textquotedblright\ as is done there.
\end{remark}

\section{Sequential decoding in the one-shot regime}

This section provides a proof for one of our main results:\ that a sequential
decoding strategy works well even in the one-shot regime. More specifically,
the theorem bounds the $\varepsilon$-one-shot classical capacity of a
classical-quantum channel, defined operationally as the maximum number of bits
that a sender can transmit to a receiver using such a channel with a failure
probability no larger than~$\varepsilon$. The general idea behind the proof is
the same as that in the proof of Theorem~1 of \citep{WR12},
with the exception that we employ a sequential decoding strategy and use
Lemma~\ref{lem:non-comm-bound}\ to bound the error probability of this
decoding strategy.

\begin{theorem}
\label{thm:seq-one-shot}A sequential decoding strategy leads to the following
bound on the $\varepsilon$-one-shot classical capacity $C^{\varepsilon}\left(
W\right)  $\ of a classical-quantum channel $W:x\rightarrow\rho_{x}$:%
\begin{equation}
C^{\varepsilon}\left(  W\right)  \geq\max_{p_{X}}D_{H}^{\varepsilon^{\prime}%
}\left(  \rho_{XB}||\rho_{X}\otimes\rho_{B}\right)  -\log_{2}\left(  \frac
{1}{\varepsilon^{2}/4-\varepsilon^{\prime}}\right)  ,
\end{equation}
for some $\varepsilon^{\prime}$ such that $\varepsilon^{2}/4>\varepsilon
^{\prime}$, where $\rho_{XB}$ is the following classical-quantum state that
depends on the distribution $p_{X}\left(  x\right)  $ and the channel $W$:%
\begin{equation}
\rho_{XB}\equiv\sum_{x}p_{X}\left(  x\right)  \left\vert x\right\rangle
\left\langle x\right\vert _{X}\otimes\left(  \rho_{x}\right)  _{B}.
\end{equation}

\end{theorem}

\begin{proof}
Fix $\varepsilon\geq0$ and a distribution $p_{X}\left(  x\right)  $. Let
$Q_{XB}$ be an operator such that $0\leq Q_{XB}\leq I_{XB}$ and
\begin{equation}
\text{Tr}\left\{  Q_{XB}\rho_{XB}\right\}  \geq1-\varepsilon^{\prime},
\end{equation}
where $\varepsilon^{\prime}$ is chosen as in the statement of the theorem. We
generate a codebook by choosing its codewords $x_{j}$\ at random, each
independently according to $p_{X}\left(  x\right)  $. Let $A_{x_{j}}$ denote
the following operator:%
\begin{equation}
A_{x_{j}}\equiv\text{Tr}_{X}\left\{  \left(  \left\vert x_{j}\right\rangle
\left\langle x_{j}\right\vert _{X}\otimes I_{B}\right)  Q_{XB}\right\}  .
\end{equation}
From Theorem~\ref{thm:naimark}, we know that to each $A_{x_{j}}$ there is
associated a qubit probe system $P_{j}$, a unitary $U_{BP_{j}}$, and a
projector $I_{B}\otimes\left\vert 1\right\rangle \left\langle 1\right\vert
_{P_{j}}$ such that for every state $\sigma$%
\begin{equation}
\text{Tr}\left\{  A_{x_{j}}\sigma\right\}  =\text{Tr}\left\{  \left(
U_{BP_{j}}^{\left(  j\right)  }\right)  ^{\dag}\left(  I_{B}\otimes\left\vert
1\right\rangle \left\langle 1\right\vert _{P_{j}}\right)  U_{BP_{j}}^{\left(
j\right)  }\left(  \sigma_{B}\otimes\left\vert 0\right\rangle \left\langle
0\right\vert _{P_{j}}\right)  \right\}  .
\end{equation}
Furthermore, it follows that for the complementary operator $I-A_{x_{j}}$, we
have the following relation:%
\begin{equation}
\text{Tr}\left\{  \left(  I-A_{x_{j}}\right)  \sigma\right\}  =\text{Tr}%
\left\{  \left(  U_{BP_{j}}^{\left(  j\right)  }\right)  ^{\dag}\left(
I_{B}\otimes\left\vert 0\right\rangle \left\langle 0\right\vert _{P_{j}%
}\right)  U_{BP_{j}}^{\left(  j\right)  }\left(  \sigma_{B}\otimes\left\vert
0\right\rangle \left\langle 0\right\vert _{P_{j}}\right)  \right\}  .
\end{equation}
(Since $\left\{  I-A_{x_{j}}, A_{x_{j}}\right\}  $ is a two-outcome POVM, the
unitary operator $U_{BP_{j}}^{\left(  j\right)  }$ can have the form given in
Example~\ref{ex:binary-POVM}.)

For a specific codebook $\left\{  x_{j}\right\}  _{j\in\left[  M\right]  }$,
the decoding strategy of the receiver Bob is as follows. Suppose that the
sender Alice wishes to transmit message $m$, so that she transmits codeword
$x_{m}$ over the channel$~W$. Then the state at the receiver is $\rho_{x_{m}}%
$. The receiver first appends $M$ ancillas, each set to $\left\vert
0\right\rangle $, to the state $\rho_{x_{m}}$\ received. Then the state at the
receiving end is as follows:%
\begin{equation}
\rho_{x_{m}}\otimes\left\vert \overline{0}\right\rangle \left\langle
\overline{0}\right\vert _{P^{M}}.
\end{equation}
Bob then checks if the codeword transmitted by Alice is the first codeword. He
does so by performing the unitary $U_{BP_{1}}^{\left(  1\right)  }$
corresponding to the first POVM\ $\left\{  I-A_{x_{1}} , A_{x_{1}}\right\}  $,
and the state becomes%
\begin{equation}
U_{BP_{1}}^{\left(  1\right)  }\left(  \rho_{x_{m}}\otimes\left\vert
\overline{0}\right\rangle \left\langle \overline{0}\right\vert _{P^{M}%
}\right)  \left(  U_{BP_{1}}^{\left(  1\right)  }\right)  ^{\dag}.
\end{equation}
He then measures the probe system $P_{1}$ in the computational basis $\left\{
\left\vert 0\right\rangle \left\langle 0\right\vert _{P_{1}},\left\vert
1\right\rangle \left\langle 1\right\vert _{P_{1}}\right\}  $. If he obtains
the outcome $\left\vert 1\right\rangle $, then he decodes that the first
message was sent (in this case, there would be an error if $m\neq1$).
Otherwise, he performs the inverse of $U_{BP_{1}}^{\left(  1\right)  }$. At
this point, if $m\neq1$ and if there is no error, the subnormalized state
becomes%
\begin{equation}
\left(  U_{BP_{1}}^{\left(  1\right)  }\right)  ^{\dag}\left(  I_{B}%
\otimes\left\vert 0\right\rangle \left\langle 0\right\vert _{P_{1}}\right)
U_{BP_{1}}^{\left(  1\right)  }\left(  \rho_{x_{m}}\otimes\left\vert
\overline{0}\right\rangle \left\langle \overline{0}\right\vert _{P^{M}%
}\right)  \left(  U_{BP_{1}}^{\left(  1\right)  }\right)  ^{\dag}\left(
I_{B}\otimes\left\vert 0\right\rangle \left\langle 0\right\vert _{P_{1}%
}\right)  \left(  U_{BP_{1}}^{\left(  1\right)  }\right)  .
\end{equation}
Making the abbreviations%
\begin{align}
I-\Pi_{x_{1}}  &  \equiv\left(  U_{BP_{1}}^{\left(  1\right)  }\right)
^{\dag}\left(  I_{B}\otimes\left\vert 0\right\rangle \left\langle 0\right\vert
_{P_{1}}\right)  U_{BP_{1}}^{\left(  1\right)  }, \label{eq:shorthand-2} \\
\Pi_{x_{1}}  &  \equiv\left(  U_{BP_{1}}^{\left(  1\right)  }\right)  ^{\dag
}\left(  I_{B}\otimes\left\vert 1\right\rangle \left\langle 1\right\vert
_{P_{1}}\right)  U_{BP_{1}}^{\left(  1\right)  },\label{eq:shorthand-1}
\end{align}
we can write the above subnormalized state as%
\begin{equation}
\left(  I-\Pi_{x_{1}}\right)  \left(  \rho_{x_{m}}\otimes\left\vert
\overline{0}\right\rangle \left\langle \overline{0}\right\vert _{P^{M}%
}\right)  \left(  I-\Pi_{x_{1}}\right)  .
\end{equation}

The receiver then continues by performing similar actions to determine if the
transmitted codeword was the second one. That is, he performs the unitary
$U_{BP_{2}}^{\left(  2\right)  }$ corresponding to $A_{x_{2}}$, measures the
probe system $P_{2}$ in the computational basis, and inverts the unitary
$U_{BP_{2}}^{\left(  2\right)  }$ if he does not receive the outcome
$\left\vert 1\right\rangle $ from the measurement of $P_{2}$.

The success probability of this sequential decoding procedure when the
$m^{\text{th}}$ codeword is sent is equal to%
\begin{equation}
\text{Tr}\left\{  \Pi_{x_{m}}\left(  I-\Pi_{x_{m-1}}\right)  \cdots\left(
I-\Pi_{x_{1}}\right)  \left(  \rho_{x_{m}}\otimes\left\vert \overline
{0}\right\rangle \left\langle \overline{0}\right\vert _{P^{M}}\right)  \left(
I-\Pi_{x_{1}}\right)  \cdots\left(  I-\Pi_{x_{m-1}}\right)  \Pi_{x_{m}%
}\right\}  .
\end{equation}
Thus, the error probability is given by%
\begin{multline}
1-\text{Tr}\left\{  \Pi_{x_{m}}\left(  I-\Pi_{x_{m-1}}\right)  \cdots\left(
I-\Pi_{x_{1}}\right)  \left(  \rho_{x_{m}}\otimes\left\vert \overline
{0}\right\rangle \left\langle \overline{0}\right\vert _{P^{M}}\right)  \left(
I-\Pi_{x_{1}}\right)  \cdots\left(  I-\Pi_{x_{m-1}}\right)  \Pi_{x_{m}%
}\right\} \\
=\text{Tr}\left\{  \left(  \rho_{x_{m}}\otimes\left\vert \overline
{0}\right\rangle \left\langle \overline{0}\right\vert _{P^{M}}\right)
\right\} \\
-\text{Tr}\left\{  \Pi_{x_{m}}\left(  I-\Pi_{x_{m-1}}\right)  \cdots\left(
I-\Pi_{x_{1}}\right)  \left(  \rho_{x_{m}}\otimes\left\vert \overline
{0}\right\rangle \left\langle \overline{0}\right\vert _{P^{M}}\right)  \left(
I-\Pi_{x_{1}}\right)  \cdots\left(  I-\Pi_{x_{m-1}}\right)  \Pi_{x_{m}%
}\right\}  .
\end{multline}
We can then upper bound this error probability by employing Lemma~\ref{lem:non-comm-bound}:
\begin{align}
&  \leq2\sqrt{\text{Tr}\left\{  \left(  I-A_{x_{m}}\right)  \rho_{x_{m}%
}\right\}  +\sum_{j=1}^{m-1}\text{Tr}\left\{  A_{x_{j}}\rho_{x_{m}}\right\}
}\\
&  \leq2\sqrt{1-\text{Tr}\left\{  A_{x_{m}}\rho_{x_{m}}\right\}  +\sum_{j\neq
m}\text{Tr}\left\{  A_{x_{j}}\rho_{x_{m}}\right\}  }.
\end{align}
Taking the expectation of the error with respect to all codebooks
(but keeping the codeword $x_m$ fixed) and
exploiting concavity of the square root function, this upper bound becomes%
\begin{equation}
2\sqrt{1-\text{Tr}\left\{  A_{x_{m}}\rho_{x_{m}}\right\}  +\left(  M-1\right)
\sum_{x^{\prime}}p_{X}\left(  x^{\prime}\right)  \text{Tr}\left\{
A_{x^{\prime}}\rho_{x_{m}}\right\}  }.
\end{equation}
Taking the expectation of the error with respect to the codeword $x_m$ itself (and
again exploiting concavity), this upper bound becomes%
\begin{equation}
2\sqrt{1-\sum_{x}p_{X}\left(  x\right)  \text{Tr}\left\{  A_{x}\rho
_{x}\right\}  +\left(  M-1\right)  \text{Tr}\left\{  \sum_{x^{\prime}}%
p_{X}\left(  x^{\prime}\right)  A_{x^{\prime}}\sum_{x}p_{X}\left(  x\right)
\rho_{x}\right\}  }. \label{eq:avg-upp-bound}%
\end{equation}
Using the facts that%
\begin{align}
\sum_{x}p_{X}\left(  x\right)  \text{Tr}\left\{  A_{x}\rho_{x}\right\}   &
=\text{Tr}\left\{  Q_{XB}\rho_{XB}\right\}  ,\\
\text{Tr}\left\{  \sum_{x^{\prime}}p_{X}\left(  x^{\prime}\right)
A_{x^{\prime}}\sum_{x}p_{X}\left(  x\right)  \rho_{x}\right\}   &
=\text{Tr}\left\{  Q_{XB}\left(  \rho_{X}\otimes\rho_{B}\right)  \right\}  ,
\end{align}
we can write the upper bound in (\ref{eq:avg-upp-bound}) as%
\begin{align}
&  2\sqrt{1-\text{Tr}\left\{  Q_{XB}\rho_{XB}\right\}  +\left(  M-1\right)
\text{Tr}\left\{  Q_{XB}\left(  \rho_{X}\otimes\rho_{B}\right)  \right\}  }\\
&  \leq2\sqrt{\varepsilon^{\prime}+M\text{Tr}\left\{  Q_{XB}\left(  \rho
_{X}\otimes\rho_{B}\right)  \right\}  }%
\end{align}
Let $R = \log_2(M)$. By optimizing the choice of the operator $Q_{XB}$ with respect to the
hypothesis testing relative entropy defined in (\ref{eq:beta-error}) and
(\ref{eq:hypo-entropy}), we find the following upper bound on the error
$\varepsilon$:%
\begin{equation}
2\sqrt{\varepsilon^{\prime}+2^{-\left[  D_{H}^{\varepsilon^{\prime}}\left(
\rho_{XB}||\rho_{X}\otimes\rho_{B}\right)  -R\right]  }}.
\end{equation}
Since we proved an upper bound on the expectation of the average error
probability with respect to the codebook choice, we can conclude that there
exists at least one code with the above bound on its average error
probability. Rewriting this upper bound on $\varepsilon$, we find that the
sequential decoding scheme gives the following bound on the $\varepsilon
$-one-shot capacity of $W$:%
\begin{equation}
C^{\varepsilon}\left(  W\right)  \geq\max_{p_{X}}D_{H}^{\varepsilon^{\prime}%
}\left(  \rho_{XB}||\rho_{X}\otimes\rho_{B}\right)  -\log_{2}\left(  \frac
{1}{\varepsilon^{2}/4-\varepsilon^{\prime}}\right)  .
\end{equation}

\end{proof}

\begin{remark}
We recover the Holevo rate for communication by considering a memoryless
classical-quantum channel and evaluating a limit as the number of channel uses
tends to infinity. We do not discuss this point any further here,
since \citep{WR12} already discussed it in detail.
\end{remark}

\begin{remark}
Of course, it is not actually necessary to use $M$ ancillas when decoding.
After performing each measurement, the receiver could store the result in a
classical memory and simply refresh a single ancilla to the state $\vert0
\rangle$.
\end{remark}

\begin{remark}
The proof of the above theorem and Lemma~\ref{lem:non-comm-bound}\ make it
clear that one can always employ Sen's bound in the error analysis for any
random coding classical communication scheme of the above form, thus serving
as a substitute for the well-known bound in Lemma~2 of
\citep{HN03}.\ Though, the performance is slightly worse than that
obtained with the Hayashi-Nagaoka bound due to the square root on the
right-hand side of Sen's bound (one can see this explicitly by comparing
Theorem~\ref{thm:seq-one-shot} with Theorem~1 of \citep{WR12}).
\end{remark}

\begin{remark}
The operation of the sequential decoder is similar in spirit to the
conditional pulse nulling receiver introduced in \citep{GHT11}\ and
experimentally implemented in \citep{CHDLG12}, in the sense that it
proceeds by performing a unitary operation, a projection, and the inverse of
the unitary for every codeword in the codebook.
\end{remark}

\begin{remark}
We can also employ sequential decoding for a task known as one-shot classical
data compression with quantum side information \citep{DW03,RR12,TH12}. In such
a task, the sender and receiver are given a classical-quantum state of the
form $\sum_{x}p_{X}\left(  x\right)  \left\vert x\right\rangle \left\langle
x\right\vert _{X}\otimes\left(  \rho_{x}\right)  _{B}$, where the sender has
the system $X$ and the receiver the system $B$. The goal is for the sender to
transmit as few classical bits as possible to the receiver, such that he can
recover the register $X$ up to a failure probability no larger than some
$\varepsilon>0$. In this case, we can show that the number of bits that need
to be sent is related to the conditional hypothesis testing entropy
\citep{DKFRR12,TH12}, defined as%
\[
H_{H}^{\varepsilon}\left(  X|B\right)  _{\rho}\equiv\max_{\sigma_{B}}%
-D_{H}^{\varepsilon}\left(  \rho_{XB}||I_{X}\otimes\sigma_{B}\right)  ,
\]
by exploiting the same kind of proof as given in \citep{RR12,TH12}%
\ combined with our proof given above. The use of sequential decoding in the
IID\ setting for this task was first done in Section~4 of \citep{WHBH12}.
\end{remark}

\subsection{Performing sequential decoding coherently}

\label{sec:coherent-sequential}We can also consider a fully coherent
implementation of the sequential decoding strategy (that is, with unitary
operations alone). For simplicity, let $\left\vert \psi\right\rangle $ denote
the state on which the coherent sequential decoding operations will act. As
before, the procedure begins by the receiver appending $M$ probe ancillas, so
that the state becomes%
\begin{equation}
\left\vert \psi\right\rangle _{B}\otimes\left\vert \overline{0}\right\rangle
_{P^{M}}.
\end{equation}
The receiver first performs the unitary $U_{BP_{1}}^{\left(  1\right)  }$
corresponding to the first codeword, leading to%
\begin{equation}
U_{BP_{1}}^{\left(  1\right)  }\left\vert \psi\right\rangle _{B}%
\otimes\left\vert \overline{0}\right\rangle _{P^{M}}.
\end{equation}
Rather than perform an incoherent projection of the probe, the receiver can
perform a controlled-NOT\ operation from the first probe system to another
ancillary system$~A_{1}$ initialized in the state~$\left\vert 0\right\rangle
_{A_{1}}$. This leads to the state:%
\begin{equation}
\left(  I_{B}\otimes\left\vert 0\right\rangle \left\langle 0\right\vert
_{P_{1}}\right)  U_{BP_{1}}^{\left(  1\right)  }\left\vert \psi\right\rangle
_{B}\otimes\left\vert \overline{0}\right\rangle _{P^{M}}\otimes\left\vert
0\right\rangle _{A_{1}}+\left(  I_{B}\otimes\left\vert 1\right\rangle
\left\langle 1\right\vert _{P_{1}}\right)  U_{BP_{1}}^{\left(  1\right)
}\left\vert \psi\right\rangle _{B}\otimes\left\vert \overline{0}\right\rangle
_{P^{M}}\otimes\left\vert 1\right\rangle _{A_{1}}.
\end{equation}
The receiver then performs the inverse unitary $(U_{BP_{1}}^{\left(  1\right)
})^{\dag}$, and by employing the shorthand in (\ref{eq:shorthand-1}) and
(\ref{eq:shorthand-2}), we can write the resulting state as%
\begin{equation}
\left(  I-\Pi_{x_{1}}\right)  _{BP_{1}}\left\vert \psi\right\rangle
_{B}\otimes\left\vert \overline{0}\right\rangle _{P^{M}}\otimes\left\vert
0\right\rangle _{A_{1}}+\left(  \Pi_{x_{1}}\right)  _{BP_{1}}\left\vert
\psi\right\rangle _{B}\otimes\left\vert \overline{0}\right\rangle _{P^{M}%
}\otimes\left\vert 1\right\rangle _{A_{1}}.
\end{equation}
Continuing a similar procedure for the second codeword leads to the expansion:%
\begin{multline}
\left(  I-\Pi_{x_{2}}\right)  _{BP_{2}}\left(  I-\Pi_{x_{1}}\right)  _{BP_{1}%
}\left\vert \psi\right\rangle _{B}\otimes\left\vert \overline{0}\right\rangle
_{P^{M}}\otimes\left\vert 00\right\rangle _{A_{1}A_{2}}\\
+\left(  \Pi_{x_{2}}\right)  _{BP_{2}}\left(  I-\Pi_{x_{1}}\right)  _{BP_{1}%
}\left\vert \psi\right\rangle _{B}\otimes\left\vert \overline{0}\right\rangle
_{P^{M}}\otimes\left\vert 01\right\rangle _{A_{1}A_{2}}\\
+\left(  I-\Pi_{x_{2}}\right)  _{BP_{2}}\left(  \Pi_{x_{1}}\right)  _{BP_{1}%
}\left\vert \psi\right\rangle _{B}\otimes\left\vert \overline{0}\right\rangle
_{P^{M}}\otimes\left\vert 10\right\rangle _{A_{1}A_{2}}\\
+\left(  \Pi_{x_{2}}\right)  _{BP_{2}}\left(  \Pi_{x_{1}}\right)  _{BP_{1}%
}\left\vert \psi\right\rangle _{B}\otimes\left\vert \overline{0}\right\rangle
_{P^{M}}\otimes\left\vert 11\right\rangle _{A_{1}A_{2}},
\end{multline}
and so forth.

\section{Gentle sequential measurements}

Winter's Gentle Operator Lemma has found numerous applications in quantum
information theory \citep{itit1999winter,ON07}.\footnote{In quantum complexity
theory, there is a similar lemma known as the \textquotedblleft almost as good
as new lemma\textquotedblright\ discovered independently by
\citep{A05}.} It states that if a two-outcome measurement has one outcome that
occurs with high probability, then the subnormalized post-measurement state is
close to the original state. More formally,

\begin{lemma}
[Gentle Operator]\label{lem:gentle}Let $\rho$ be a state, and let $\Lambda$ be
an operator such that $0\leq\Lambda\leq I$. Then%
\begin{equation}
\left\Vert \rho-\sqrt{\Lambda}\rho\sqrt{\Lambda}\right\Vert _{1}\leq
2\sqrt{\text{\emph{Tr}}\left\{  \left(  I-\Lambda\right)  \rho\right\}  }.
\end{equation}
Thus, if \emph{Tr}$\left\{  \Lambda\rho\right\}  \geq1-\varepsilon$ for some
small $\varepsilon\geq0$, then%
\begin{equation}
\left\Vert \rho-\sqrt{\Lambda}\rho\sqrt{\Lambda}\right\Vert _{1}\leq
2\sqrt{\varepsilon}.
\end{equation}

\end{lemma}

Of course, this lemma can be extended with the Naimark extension theorem as
well. Suppose that we know that%
\begin{equation}
\text{Tr}\left\{  \Lambda\rho\right\}  \geq1-\varepsilon,
\end{equation}
for a two-outcome POVM $\left\{  \Lambda,I-\Lambda\right\}  $. By
Theorem~\ref{thm:naimark}, we know that there exists a unitary $U_{SP}$ such
that for the orthonormal basis $\left\{  \left\vert 0\right\rangle
_{P},\left\vert 1\right\rangle _{P}\right\}  $, we have that%
\begin{equation}
\text{Tr}\left\{  \Pi_{SP}\left(  \rho_{S}\otimes\left\vert 0\right\rangle
\left\langle 0\right\vert _{P}\right)  \right\}  =\text{Tr}\left\{
\Lambda\rho\right\}  ,
\end{equation}
where%
\begin{equation}
\Pi_{SP}\equiv U_{SP}^{\dag}\left(  I_{S}\otimes\left\vert 0\right\rangle
\left\langle 0\right\vert _{P}\right)  U_{SP}.
\end{equation}
Thus if we perform the unitary $U_{SP}$, the projective measurement $\left\{
\left\vert 0\right\rangle \left\langle 0\right\vert _{P},\left\vert
1\right\rangle \left\langle 1\right\vert _{P}\right\}  $, followed by the
inverse unitary $U_{SP}^{\dag}$, we can conclude from Lemma~\ref{lem:gentle}%
\ that the resulting state of both the system and the probe is close to the
original state:%
\begin{equation}
\left\Vert \rho_{S}\otimes\left\vert 0\right\rangle \left\langle 0\right\vert
_{P}-\Pi_{SP}\left(  \rho_{S}\otimes\left\vert 0\right\rangle \left\langle
0\right\vert _{P}\right)  \Pi_{SP}\right\Vert _{1}\leq2\sqrt{\varepsilon}.
\end{equation}
So, for simplicity, in what follows, we just consider the state $\rho$ to be
the state of the combined system and any necessary ancillas so that we can
consider projective measurements only (this is due to the above observation
and the Naimark extension theorem).

In a sequential decoding scheme, we also might like to conclude that the state
after the decoding procedure is close to the original state. This would be
pleasing conceptually and would also have applications in constructing
decoders for quantum data from decoders for classical data
\citep{ieee2005dev,RB08,WR12a}. Though, as noted in \citep{S11}, we cannot
generally make the above conclusion. Here, we show how performing additional
operations leads to a state close to the original one.

There are at least two ways that we can perform additional operations in order
to guarantee that the sequentially decoded state is close to the original one.
The first was mentioned at the end of Section~4.3\ of \citep{WHBH12}\ and
relies on the polar decomposition. Given that the post-measurement state is of
the following form (omitting normalization):%
\begin{equation}
\Pi_{N}\cdots\Pi_{1}\rho\Pi_{1}\cdots\Pi_{N},
\end{equation}
the receiver could perform a unitary $V$\ given by the polar decomposition%
\begin{equation}
\sqrt{\Pi_{1}\cdots\Pi_{N}\cdots\Pi_{1}}=V\Pi_{N}\cdots\Pi_{1}%
\label{eq:polar-decomp}%
\end{equation}
so that the post-measurement state becomes%
\begin{equation}
\sqrt{\Pi_{1}\cdots\Pi_{N}\cdots\Pi_{1}}\rho\sqrt{\Pi_{1}\cdots\Pi_{N}%
\cdots\Pi_{1}}.
\end{equation}
In this case, we can apply the Gentle Operator Lemma (Lemma~\ref{lem:gentle})
to upper bound the disturbance:%
\begin{align}
&  \left\Vert \rho-\sqrt{\Pi_{1}\cdots\Pi_{N}\cdots\Pi_{1}}\rho\sqrt{\Pi
_{1}\cdots\Pi_{N}\cdots\Pi_{1}}\right\Vert _{1}\nonumber\\
&  \leq2\sqrt{\text{Tr}\left\{  \left(  I-\Pi_{1}\cdots\Pi_{N}\cdots\Pi
_{1}\right)  \rho\right\}  }\\
&  =2\sqrt{1-\text{Tr}\left\{  \Pi_{N}\cdots\Pi_{1}\rho\Pi_{1}\cdots\Pi
_{N}\right\}  },
\end{align}
and then once again apply Sen's non-commutative union bound
(Theorem~\ref{thm:sen}) to upper bound the disturbance as%
\begin{equation}
2\sqrt{2}\ \sqrt[4]{\sum_{i=1}^{N}\text{Tr}\left\{  \left(  I-\Pi_{i}\right)
\rho\right\}  }.
\end{equation}
(In the context of classical-quantum polar codes, a quantity like the above
will be exponentially small in the number of channel uses because each term is
exponentially small while there are only a linear number of terms \citep{WG11}.)

One practical problem with the above approach is as follows. Suppose that we
can efficiently implement each of the measurements corresponding to the
projections $\Pi_{1}$, \ldots, $\Pi_{N}$ (say, on a quantum computer). Then we
can clearly perform the sequential decoding procedure efficiently if $N$ is
not too large. On the other hand, given a particular sequence of measurements,
it is not clear at all that the unitary given by the polar decomposition in
(\ref{eq:polar-decomp}) has an efficient implementation. Thus, the above
approach does not realize both desirable requirements of having a small
disturbance and an efficient sequential decoding (if each of the measurements
can be efficiently implemented to begin with).

There is a simple way to remedy the aforementioned problem if the sequential
decoding strategy has a very small error probability. We can simply perform
the projections $\Pi_{1}$ through $\Pi_{N}$ and then perform them again in the
opposite order. This gives the subnormalized post-measurement state%
\begin{equation}
\Pi_{1}\cdots\Pi_{N}\cdots\Pi_{1}\rho\Pi_{1}\cdots\Pi_{N}\cdots\Pi_{1},
\end{equation}
and the error probability is bounded as follows, again by applying Sen's bound
(Theorem~\ref{thm:sen}):
\begin{equation}
\text{Tr}\left\{  \rho\right\}  -\text{Tr}\left\{  \Pi_{1}\cdots\Pi_{N}%
\cdots\Pi_{1}\rho\Pi_{1}\cdots\Pi_{N}\cdots\Pi_{1}\right\}  \leq2\sqrt
{2\sum_{i=1}^{N}\text{Tr}\left\{  \left(  I-\Pi_{i}\right)  \rho\right\}
}.\label{eq:using-sen-bound}%
\end{equation}
Thus, by performing the measurements again in reverse, we only increase the
error probability by a factor of $\sqrt{2}$. Furthermore, the Gentle Operator
Lemma (Lemma~\ref{lem:gentle}) gives the following upper bound on the
disturbance:%
\begin{align}
&  \left\Vert \rho-\Pi_{1}\cdots\Pi_{N}\cdots\Pi_{1}\rho\Pi_{1}\cdots\Pi
_{N}\cdots\Pi_{1}\right\Vert _{1}\nonumber\\
&  \leq2\sqrt{\text{Tr}\left\{  \left(  I-\left[  \Pi_{1}\cdots\Pi_{N}%
\cdots\Pi_{1}\right]  ^{2}\right)  \rho\right\}  }\\
&  =2\sqrt{1-\text{Tr}\left\{  \Pi_{1}\cdots\Pi_{N}\cdots\Pi_{1}\rho\Pi
_{1}\cdots\Pi_{N}\cdots\Pi_{1}\right\}  }.
\end{align}
Applying the bound in (\ref{eq:using-sen-bound}) gives the following upper
bound on the disturbance:%
\begin{equation}
2\sqrt{2}\left(  \sqrt[4]{2}\right)  \sqrt[4]{\sum_{i=1}^{n}\text{Tr}\left\{
\left(  I-\Pi_{i}\right)  \rho\right\}  }.
\end{equation}
Thus, with this scheme, we can realize both requirements of having an
efficient implementation and a small disturbance---the receiver simply has to
perform $2N$ measurements (each of which were assumed to have an efficient
implementation) while the disturbance increases only by a factor of
$\sqrt[4]{2}$. An efficient coherent implementation of these operations
follows from the discussion in Section~\ref{sec:coherent-sequential}, if each
measurement has an efficient implementation.

By the methods of \citep{WR12a}, this latter approach will be useful for
decoding quantum polar codes, should a method ever be found to realize an
efficient decoder for classical-quantum polar codes \citep{WG11} (see
\citep{WLH13} for progress in this direction). At the very least, this
latter approach answers an open question from \citep{WR12a}.

\section*{Acknowledgements}

I am grateful to Nilanjana Datta and Renato Renner for organizing the
\textquotedblleft Beyond i.i.d.~in information theory\textquotedblright%
\ workshop in Cambridge, UK, where some of the questions in this work were
raised. I also thank them and Aram Harrow, Patrick Hayden, Olivier
Landon-Cardinal, Ligong Wang, and Andreas Winter for helpful discussions. I
acknowledge support from the Centre de Recherches Math\'{e}matiques in
Montreal and am grateful for the hospitality of the Center for Mathematical
Sciences at the University of Cambridge and the Institute for Theoretical
Physics at ETH Zurich during a research visit in January and February of 2013,
when some of the work for this paper was completed.


\end{document}